\documentclass[11pt]{article}
\pdfoutput=1

\usepackage{a4wide}
\usepackage{Packages}

\usepackage{Definitions}

\newcommand{\OfficialTitle}{On domain wall boundary conditions for the
  \textsc{xxz} spin Hamiltonian}

\hypersetup{pdfauthor={Domenico Orlando, Susanne
    Reffert, Nicolai Reshetikhin},pdftitle={\OfficialTitle}}

\author{
  \begin{minipage}{.97\linewidth}
    \vspace{1cm}
    \begin{center}
      \begin{small}
        \textbf{Domenico Orlando}${}^1$, \textbf{Susanne
          Reffert}${}^{1}$ and \textbf{Nicolai Reshetikhin}${}^{2,3}$
      \end{small}
    \end{center}
    \vspace{1cm} \hspace{2cm}\begin{minipage}{.7\linewidth}
      {\it \begin{footnotesize}
          \begin{itemize}
          \item[${}^1$] Institute for the Mathematics and Physics of
            the Universe, \\The University of Tokyo, Kashiwa-no-Ha
            5-1-5, \\ Kashiwa-shi, 277-8568 Chiba, Japan.
          \item[${}^2$] Department of Mathematics, University of California at Berkeley, \\
          Berkeley, CA 94720-3840, USA.
          \item[${}^3$] KdV Institute for Mathematics, University of Amsterdam,\\
            Plantage Muidergracht 24, 1018 TV Amsterdam, The Netherlands.
          \end{itemize}
        \end{footnotesize}}
    \end{minipage}
    \vspace{1cm}
  \end{minipage}
}

\date{}

\title{\vspace{1.5cm}
  \begin{huge}
    \textbf{\OfficialTitle}
  \end{huge}
}

\begin{document}

\numberwithin{equation}{section}

\begin{titlepage}
  \maketitle
  \thispagestyle{empty}

  \vspace{-14cm}
  \begin{flushright}
    IPMU09-0137
  \end{flushright}

  \vspace{14cm}

  \begin{center}
    \textsc{Abstract}\\
  \end{center}
  In this note, we derive the spectrum of the infinite
  quantum \textsc{xxz} spin chain with domain wall boundary conditions.
  The
  eigenstates are constructed as limits of  Bethe states for
   the finite \textsc{xxz} spin chain with
  $U_q(sl_2)$ invariant boundary conditions.
\end{titlepage}


\tableofcontents

\section{Introduction}\label{sec:intro}

In this paper we consider the \textsc{xxz} spin chain with domain wall (DW)
boundary conditions. In physical terms it describes a one--dimensional
lattice where each point carries a spin interacting with its
neighbors. Its time evolution is given by the Hamiltonian
\begin{equation}
  \label{eq:xxz}
  \HH = - \frac{1}{2} \sum_{k\in \ZZ+1/2} (\sigma^1_k \sigma_{k+1}^1 +  \sigma^2_k \sigma_{k+1}^2 + \Delta \ \sigma^3_k \sigma_{k+1}^3  - \Delta )\, ,
\end{equation}
where $k$ denotes the position on the lattice, $\sigma^i_k$ is the
$i$--th Pauli matrix acting on the spin at position $k$, and $\Delta$
is the parameter characterizing the anisotropy of the spin
interaction. In our case, $\Delta >1$.

This is one of the most studied systems in statistical physics. It has
a number of remarkable properties. In particular it is integrable,
i.e. the Hamiltonian (\ref{eq:xxz}) can be diagonalized using the
Bethe ansatz~\cite{Bethe:1931hc}.

In this paper we study the spectrum of \textsc{xxz} Hamiltonian in
the space of states with domain wall boundary conditions.
This space is the natural $l_2$-completion of
\begin{equation}
\mathrm{span} \set{\prod_i\sigma^-_{x_i}\prod_i\sigma^+_{y_i} \Omega_{DW} | x_i < 0, y_i > 0  }\,,
\end{equation}
where $\Omega_{DW} = \ket{\cdots\ua\ua\da\da\cdots}$.

This paper is the continuation of~\cite{Dijkgraaf:2008},
where it was
shown that the evolution of a system of random integer partitions can
be mapped to the half--filled sector of the ferromagnetic \textsc{xxz}
spin chain with domain wall (or kink) boundary conditions.
In~\cite{Dijkgraaf:2008}, the quantization of the system of integer
partitions or Young diagrams was considered.  There is a well--known
map from partitions to the \textsc{ns} sector of a one--dimensional
free fermion. A fermionic creation operator $\psi^*_{a}$ corresponds
to each black dot of the \emph{Maya diagram}, \ie the projection on
the horizontal line (see Fig.~\ref{fig:2d}).  The empty partition
$\Omega_{DW}$ is obtained by filling all negative positions. Any other
partition can be obtained by acting with creator--annihilator pairs on
$\Omega_{DW}$.  In a quantization procedure familiar from the theory
of quantum dimers (see~\cite{Rokhsar1988,Dijkgraaf:2009gr}), the
Hilbert space of the quantum system is spanned by vectors that are in
one--to--one correspondence to the classical configurations (the
integer partitions). In terms of spin chains, this Hilbert space is
precisely $\mathscr{H}_{DW}$. The ground state $B_\infty$ of the
quantum Hamiltonian is required to reproduce the steady state
distribution of the classical system.  A natural Hamiltonian
fulfilling this requirement is given by
\begin{equation}
  \label{eq:Hamil2d}
  \HH = - \sum_{k \in \setZ +
      1/2}\psi^*_{k+1}\psi_k+\psi^*_k\psi_{k+1} - q\,n_k \left( 1
      - n_{k+1} \right) - \frac{1}{q}\,n_{k+1} \left( 1 - n_k
    \right) \, ,
\end{equation}
where $n_k=\psi_k^*\psi_k$ is the fermion number operator. This
Hamiltonian can be recast into the form of Eq.~\eqref{eq:xxz} by using
the Jordan--Wigner transformation. The quantization of the growth of
integer partitions thus corresponds to the \textsc{xxz} spin chain
with domain wall boundary conditions.

\begin{figure}
  \begin{center}
    \includegraphics[width=.5\textwidth]{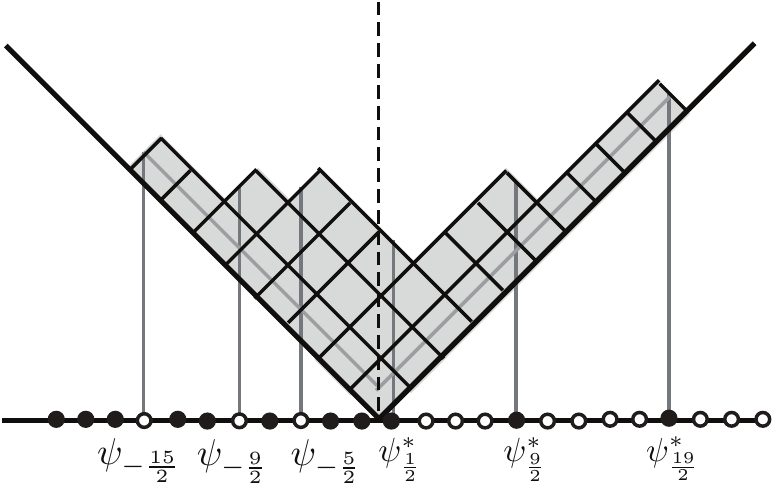}
  \end{center}
  \caption{Two--dimensional partition corresponding to the state
    $\Omega_{DW}([10,6,3,3,3,2,1,1]) = \psi_{\frac{19}{2}}^* \psi_{-\frac{15}{2}}\psi_{\frac{9}{2}}^*
    \psi_{-\frac{9}{2}}\psi_{\frac{1}{2}}^* \psi_{-\frac{5}{2}}
    \Omega_{DW} $.}
  \label{fig:2d}
\end{figure}
%

The quantum \textsc{xxz} model is also related to a classical random
process system, to the \emph{asymmetric exclusion process} (\textsc{asep}).
This is a \emph{Markov process} on a one dimensional lattice with particles
on its sites. The particles diffuse with different probability rates
for jumping left or right. None of the sites in such a system can be
occupied by more than one particle.
The \emph{transition rate matrix} $W$ for such a process can be obtained by
conjugating the \textsc{xxz} Hamiltonian above with a diagonal
matrix.

The gap in the excitation spectrum has been
determined in~\cite{Nachtergaele}. We construct the \emph{whole spectrum}.
Because the transition rate matrix is symmetrizable, the \textsc{asep}
is invertible and therefore it converges to the equilibrium
state which is the ground state of the \textsc{xxz} Hamiltonian over
the DW space of states.

The problem of computing matrix elements (form factors)
  \begin{equation}
    \left(\psi_\alpha,\sigma^{\alpha_1}_{x_1}\cdots\sigma^{\alpha_k}_{x_n}\psi_\beta\right)
  \end{equation}
for all eigenvectors $\psi_\alpha, \psi_\beta$ of $\HH$ remains.

Our main results are the following: we describe the eigenvectors and
eigenvalues of the \textsc{xxz} Hamiltonian on the DW Hilbert space and we find explicit expressions for the form factors of the ground state extending results of~\cite{Dijkgraaf:2008}. We also study the $q \to 1^-$ limit of these (scaling) matrix elements.

\bigskip

The plan of this paper is as follows. In
Section~\ref{sec:xxz}, the recall some basic facts about \textsc{xxz}
Hamiltonian, and the asymmetric exclusion process (\textsc{asep}) is
introduced. The main part of the paper is Section~\ref{sec:hamil}
where we describe the spectrum of the \textsc{xxz} model with DW boundary conditions.
In Sec.~\ref{sec:asep}, the time
evolution for the \textsc{asep} is treated. In Sec.~\ref{sec:limit-shapes},
some calculations are performed on the ground state, such as
the limit shape and some matrix elements.

\section{The \textsc{xxz} Hamiltonian}\label{sec:xxz}

Here, we recall the definitions of the DW Hilbert space and
the \textsc{asep} Markov process.

\subsection{The DW Hilbert space}
\label{sec:hilbert}

Denote by $\Omega$ the \emph{ferromagnetic ground state}
with all spins up, and by $\Omega(x_1, \dots, x_m) $ the state where the spins at positions $x_1< x_2< \dots$ are down.

The \emph{domain wall} (DW) ground state $\Omega_{DW}$ is the state where all the spins with negative coordinates are up and all the spins with positive coordinates are down. The states
\begin{equation}
\prod_{i=1}^n\sigma^-_{u_i}\prod_{i=1}^m\sigma^+_{v_i} \Omega_{DW} \end{equation}
form a basis in the DW Hilbert space.
When $n=m$, these vectors can be parametrized by partitions as $\Omega_{DW}(\delta)$, where $\delta = [\delta_1, \dots, \delta_r ]\in \mathscr{S}$ is the integer partition which has $\mathbf{u}$ and $\mathbf{v}$ as modified Frobenius coordinates:
\begin{align}
  v_i = \delta_i - i + \frac{1}{2} && u_i = - \left( \delta_i^t - i + \frac{1}{2} \right) \, ,
\end{align}
where $i$ ranges from $1$ to the number $m$ of the squares on the diagonal of the Young diagram of $\delta$. The partition $\delta^t$ is the reflection of $\delta$ with respect to the diagonal.

\subsection{The \textsc{asep} evolution of the DW states}
\label{sec:asep_el}

\subsubsection{Continuous Markov process} Recall that a continuous--time Markov process is defined by the following evolution equation:
\begin{equation}
  \frac{\di P_C (t)}{\di t} = \sum_{C^\prime } W(C, C^\prime ) P_{C^\prime} (t) \, ,
\end{equation}
where $C$ is a configuration of the system, $P_C (t)$ is the
probability that the system is in configuration $C$ at time $t$, and $W
(C, C^\prime)$ is the transition rate from $C^\prime $ to $C$.
Because of $\sum_C P_C (t) = 1$, and that $P_C(t)>0$, one finds that $W(C, C^\prime)$ must
satisfy the following conditions:
\begin{align}
  W (C, C^\prime ) > 0 \hspace{2em} \forall C^\prime \neq C && \text{(positivity),}\label{eq:cond1} \\
  W( C, C ) = - \sum_{C^\prime \neq C} W( C^\prime, C ) && \text{(total probability).\label{eq:cond2}}
\end{align}
Here we will focus on Markov processes satisfying the detailed
balance condition (so--called invertible Markov processes).
The transition rates in such processes satisfy the symmetry relation
\begin{equation}
  W ( C^\prime, C ) P^0_C = W  (C, C^\prime ) P^0_{C^\prime} \, ,
\end{equation}
where $P^0_C>0$ and we choose them to be normalized as
$\sum_C P_C^0 = 1$. The symmetry of the transition matrix implies that $\{P^0_C\}$ is an eigenvector of $W$ with  eigenvalue $0$:
\begin{equation}
  \sum_{C^\prime} W (C, C^\prime ) P^0_{C^\prime} = 0 \, ,
\end{equation}
One of the fundamental features of invertible Markov processes is that
such a process starting from any distribution $P(0)$ converges to
$P^0$ as $t\to \infty$.

\subsubsection{Asymmetric exclusion process} An asymmetric exclusion process (\textsc{asep}) is a
Markov process on the infinite spin chain. In this paper we will focus
on the study of such processes in the space $\Hi_{DW}$.

The transition rate matrix for the \textsc{asep} is determined by its
off--diagonal matrix elements which can be described loosely as
\begin{align}
  W ( \set{\dots \ua \da \dots }, \set{\dots \da \ua \dots } ) = w^- ,&&
  W ( \set{\dots \da \ua \dots }, \set{\dots \ua \da \dots } ) = w^+ \, .
\end{align}
This linear operator can be written as
  \begin{equation}
  W=\sum_{k\in \ZZ} \left( w^-\, \sigma^+_k \sigma_{k+1}^- + w^+\,
  \sigma^-_k \sigma_{k+1}^+ + \frac{w^+ + w^-}{4} \left( \sigma^3_k
    \sigma_{k+1}^3 - 1  \right) \right) \,.
\end{equation}

Let us parametrize the weights $w^\pm$ by
$w^\pm = A\, q^{\pm 1}$. Then it is clear that the
transition rate matrix for the \textsc{asep} has the symmetry
\begin{equation}
  W(A,q^{-1})=C W(A,q) C \, ,
\end{equation}
where
\begin{equation}
  C=\bigotimes_{i\in \ZZ}
  \begin{pmatrix}
    0 & 1 \\ 1 & 0
  \end{pmatrix} .
\end{equation}
Because of this symmetry we will assume $q<1$
for the rest of the paper.

Define the linear operator
\begin{equation}
  U=\left(\bigotimes_{i<0} \begin{pmatrix}
      1 & 0 \\ 0 & q^{-i}
    \end{pmatrix} \right)\otimes \left(\bigotimes_{i>0} \begin{pmatrix}
      q^i & 0 \\ 0 & 1
    \end{pmatrix} \right).
\end{equation}
It acts on the space $\Hi_{DW}$ as
\begin{equation}
  U \prod_{i=1}^m \sigma_{x_i}^- \prod_{j=1}^n\sigma_{y_i}^+
  \Omega_{DW}  = q^{\sum_i \left( y_i - x_i \right) }  \prod_{i=1}^m
  \sigma_{x_i}^- \prod_{j=1}^n\sigma_{y_i}^+\Omega_{DW} \, ,
\end{equation}
and it is a bounded operator acting on this Hilbert space.
It is easy to check that
\begin{equation}
  U \sigma_i^\pm U^{-1} = q^{\pm i} \sigma_i^\pm \, .
\end{equation}
These identities have two important implications:
\begin{equation}
  W = - A U H U^{-1}, \ \ W = U^2 W^t U^{-2} \, ,
\end{equation}
where $W^t$ is the transposed matrix $W$ and
\begin{equation}
  H = - \sum_{k\in \ZZ} \left( \sigma^+_k \sigma_{k+1}^- +  \sigma^-_k
  \sigma_{k+1}^+ + \frac{\Delta }{2} \left( \sigma^3_k \sigma_{k+1}^3
    - 1 \right) \right) \, .
\end{equation}
Here $\Delta=\tfrac{1}{2}(q+q^{-1})$. Therefore, the matrix $W$ is
conjugate to the \textsc{xxz} Hamiltonian and it satisfies the
\emph{detailed balance condition}. This implies that the vector
\begin{equation}
  \label{eq:ground}
  \Psi_0= \frac{1}{Z}\sum_{n\geq 0}\sum_{x_1<x_2<\dots x_n<0<y_1<\dots<y_n}\,q^{\sum_{i=1}^n(y_i-x_i)}\prod_{i=1}^m \sigma_{x_i}^- \prod_{j=1}^n\sigma_{y_i}^+ \Omega_{DW}
\end{equation}
is a normalized eigenvector of $W$ with eigenvalue $0$
and that the \textsc{asep} evolution with the transition rate $W$
converges to $\Psi_0$ for any initial condition as $t\to \infty$.

\section{The Bethe ansatz}\label{sec:hamil}

In this section, we construct the full spectrum of the Hamiltonian on
the $U_q(sl_2)$--symmetric \textsc{xxz} chain.

\subsection{The Bethe ansatz for the \textsc{xxz} chain over the ferromagnetic vacuum}\label{sec:bethe}

The Bethe ansatz~\cite{Bethe:1931hc} for the \textsc{xxz} spin chain
describes explicitly the eigenvectors and eigenvalues of the
Hamiltonian in Eq.~\eqref{eq:xxz} over the ferromagnetic vacuum
$\Omega$. More specifically, consider the vector
\begin{equation}
  B(\mathbf{z}) = {\sum_{ x_1 < \dots < x_m}}  \lambda ( \mathbf{x} |
  \mathbf{z} ) \Omega(\mathbf{x}) \, ,
\end{equation}
where
\begin{equation}
\label{eq:general-lambda}
  \lambda( \mathbf{x} |\mathbf{z}) = \sum_{\pi\in S_m} \epsilon_{\pi}
  A(z_{\pi(1)}, \dots, z_{\pi(m)}) z_{\pi(1)}^{x_1}, \dots,
  z_{\pi(m)}^{x_m} = \sum_{\pi\in S_m} \epsilon_{\pi}
  A(\mathbf{z}_{\pi}) \mathbf{z}_{\pi}^{\mathbf{x}} \, .
\end{equation}
Here $\abs{z_i} = 1$, the sum runs over the permutations of $\set{z_1, \dots,
  z_m}$, and $\epsilon_\pi$ is the sign of $\pi$.
The explicit expression for $A(\mathbf{z})$ is
\begin{equation}
  A (z_1,\ldots, z_m ) = \prod_{i<j}A(z_i,z_j)=\prod_{i<j} \left( 1 - 2\,\Delta\, z_i + z _iz_j \right) \, .
\end{equation}
Bethe proved that this vector is an eigenvector of~(\ref{eq:xxz}) with
the eigenvalue
\begin{equation}
  \mathcal{E} (\mathbf{z})  =  \sum_{a=1}^m \left( 2 \Delta - z_a - z_a^{-1}\right) \, .
\end{equation}
Such eigenvectors are called $m$--magnon states.

In sectors with $m\geq2$, the magnons can form bound states. These
bound states decay exponentially with the distance between the
magnons. In the eigenvector corresponding to a bound state of
$n$ magnons, the numbers $( z_1, \dots, z_n ) \in ( \setC^*)^n$ are such
that $\abs{z_1 \dots z_n} = 1$, $\abs{z_i} > \abs{z_{i+1}}$, and $z_i
= h^{n-i}(z_n)$, where $h(z_n)= 2\Delta - \frac{1}{z_n}$.  The
corresponding eigenvector is
 \begin{equation}
  B(z_1, \dots, z_n) = \sum_{x_1<x_2< \dots < x_n} A(z_1, \dots, z_n)
  \left( z_1 \dots z_n \right)^{\frac{x_1 + \dots + x_n}{n}} \prod_{1
    \le i < j \le n} \left( \frac{z_j}{z_i} \right)^{\frac{x_j -
      x_i}{n}} \Omega(x_1, \dots, x_n) \, .
\end{equation}
It is convenient to use the parametrization of $z_j$ by \emph{rapidities},
see \emph{e.g.}~\cite{gaudin}:
\begin{align}
  z_j=\frac{\sh(i\beta+(\frac{n}{2}-j+1)\eta)}{\sh(i\beta+(\frac{n}{2}-j)\eta)}
  \, , &&
 z_1 \dots
 z_n = \frac{\sh(i\beta+\frac{n\eta}{2})}{\sh(i\beta-\frac{n\eta}{2})}
 \, ,
\end{align}
where $\Delta = \cosh (\eta)$. In this parametrization, the energy of a bound state of $n$ magnons
is
\begin{equation}
  \mathcal{E}_n(\beta)= \frac{ 2 \sh(\eta) \, \sh (n\eta)}{\ch(n \eta)
      - \cos(2 \beta) }
  \, .
\end{equation}

The eigenvectors with an arbitrary number of magnons and their bound
states form the complete continuous spectrum of the \textsc{xxz} spin
chain over the ferromagnetic vacuum.


\subsection{The spectrum of the \textsc{xxz} Hamiltonian over the DW vacuum}
\label{sec:spectrum}

The Bethe states above do not belong to the Hilbert space $\mathscr{H}_{DW}$.
It is however possible to map them to a complete set of solutions of
the \textsc{xxz} Hamiltonian for the infinite chain with domain wall boundary
conditions.

Our main result is the following. The states $B_\infty (\mathbf{z} )$ defined by
\begin{equation}
  \label{eq:B-result}
  B_{\infty}(\mathbf{z}) = q^{\binom{m+1}{2}} \sum_{\delta \in \mathscr{S}} q^{\abs{\delta}} \sum_{1 \leq l_1 < \dots < l_m < \infty} \frac{\lambda ( \delta_{l_1} - l_1, \dots, \delta_{l_m} - l_m | \mathbf{z}) }{q^{\sum_{a=1}^m \delta_{l_a}+l_a}} \Omega_{DW}(\delta)
\end{equation}
form a complete set of eigenstates of the continuum spectrum for the
\textsc{xxz} Hamiltonian for the infinite chain with domain wall
boundary conditions with eigenvalues
\begin{equation}
  \label{eq:ev}
  \mathcal{E}(\mathbf{z}) =  \sum_{a=1}^m \left( 2 \Delta - z_a - z_a^{-1} \right) \, .
\end{equation}
Here, the $\mathbf{z}$ are as above, $\mathscr{S}$ is the set of
integer partitions, $\abs{\delta} = \delta_1 + \dots + \delta_r$, and
the $\lambda(\mathbf{x}|\mathbf{z})$ are defined in
Eq.~(\ref{eq:general-lambda}).

The immediate corollary of this is that the vectors
\begin{equation}
  \Omega_\infty(\mathbf{z}) = U B_{\infty} (\mathbf{z}) =  q^{\binom{m+1}{2}} \sum_{\delta \in \mathscr{S}} q^{2 \abs{\delta}} \sum_{1 \leq l_1 < \dots < l_m < \infty} \frac{\lambda ( \delta_{l_1} - l_1, \dots, \delta_{l_m} - l_m | \mathbf{z}) }{q^{\sum_{a=1}^m \delta_{l_a}+l_a}} \Omega_{DW}(\delta)
\end{equation}
form a complete set of eigenstates for the rate matrix $W = A U \HH
U^{-1}$ of the \textsc{asep} process.

\bigskip
The strategy of the proof is to pass to the limit $N\to \infty$ from
the spectrum of the \textsc{xxz} Hamiltonian for finite spin chain
with $U_q(sl_2)$ invariant boundary condition. The latter was computed
in~\cite{Pasquier:1990}. Let $B^{(N)}(\mathbf{z})$ be a Bethe vector
for the spin chain of length $N$ (see Appendix B). The Hamiltonian is
invariant with respect to the diagonal action of $U_q(sl_2)$,
therefore the vectors

\begin{equation}
  B_p^{(N)} (\mathbf{z}) =\frac{q^{\binom{p}{2}}}{[ p ]_{q^2}!}
  K^{-\nicefrac{p}{2}} F^p B^{(N)} (\mathbf{z}) \, ,
\end{equation}
are also eigenvectors of such spin chain. Here
\begin{align}
  [p]_{q^2} = \frac{1 - q^{2p}}{1-q^2}\, ,&& [p]_{q^2}! = [p]_{q^2} [p-1]_{q^2} \dots [2]_{q^2} [1]_{q^2} \, .
\end{align}
and $K$ and $F$ are the diagonal actions of the generators of
$U_q(sl_2)$ (see appendix A) on ${\CC^2}^{\otimes N}$:
\begin{gather}
  K = q^{\sigma^3} \otimes \dots \otimes q^{\sigma^3} \, ,\\
  \label{eq:F-generator}
  F=\sum_y\underbrace{q^{\sigma^3}\otimes\dots \otimes q^{\sigma^3}}_{y-1}\otimes\, \sigma^-\otimes\mathbbm{1}\otimes\dots\otimes\mathbbm{1} = \sum_yL_y\,\sigma_y^- \, .
\end{gather}

\begin{lemma} We have the following identity:
\begin{multline}
  \frac{q^{\binom{p}{2} }}{[ p ]_{q^2}!} K^{-\nicefrac{p}{2}} F^p
  \Omega(x_1,\dots,x_m) = \sum_{\substack{\frac{N}{2}-\frac{1}{2} >
      y_1 > \dots > y_p > -\frac{N}{2}+\frac{1}{2} \\ y_i\neq x_a}}
  q^{p \left( 2p + 2m - 1 \right) + \sum_{i=1}^p \left( y_i - 2 l(y_i|\mathbf{x})  \right)}\times\\
  \times \Omega( \set{x_a}_1^m,\,\set{y_i}_1^p) \, ,
\end{multline}
where $l(y_i|\mathbf x)=\#\{x_a|x_a<y_i\}$.
\end{lemma}
\begin{proof}
The commutation relations between $\sigma^-$ and $L_y$ are given by
\begin{equation}
  \sigma^-_{y_1} L_{y_2}=
  \begin{cases}
    L_{y_2}\sigma^-_{y_1},\quad & \text{if $y_1 \ge y_2$}\\
    q^{2} L_{y_2} \sigma^-_{y_1}\quad & \text{if $y_1 <  y_2\,.$}
  \end{cases}
\end{equation}
From here we find:
\begin{equation}
  \label{eq:prefact}
  \begin{split}
    F^p &=\sum_{y_1,\dots,y_p} L_{y_1}\sigma^-_{y_1}
    L_{y_2}\sigma^-_{y_2} \dots L_{y_p}\sigma^-_{y_p} =\sum_{\pi\in
      S_p}q^{2l(\pi)} \sum_{y_1 > y_2,\dots > y_p}L_{y_1}L_{y_2}\dots
    L_{y_p}\sigma^-_{y_1}\dots\sigma^-_{y_p} \\
    &= [p]_{q^2}! \sum_{y_1 > y_2,\dots > y_p}L_{y_1}L_{y_2}\dots L_{y_p}\sigma^-_{y_1}\dots\sigma^-_{y_p},
  \end{split}
\end{equation}
where $\pi$ is a permutation in the symmetric group $S_p$, and $l(\pi)$ is
the \emph{length} of the permutation. With this,
\begin{equation}
  F^p \Omega(x_1,\dots,x_m) = [p]_{q^2}! \sum_{\substack{ y_1 > \dots > y_p \\ y_i\neq x_a}}L_{y_1}\dots L_{y_p}\Omega(\set{x_a}_1^m,\,\set{y_i}_1^p ),
\end{equation}
where $\set{x_a}_1^m$ denotes the positions of the $m$ particles of the
original state and $\set{y_i}_1^p$ denotes the set of positions of the
spins flipped by the action of the $\sigma^-$. The lemma follows after evaluating the action of  $L_{y_i}$.
\end{proof}

It is convenient to introduce new coordinates. Assume that $y_i<y_{i+1}$, $x_i<x_{i+1}$, $y_{l_1-1}<x_1<y_{l_1}$,
and $y_{l_1+\dots+l_a-1}<x_a<y_{l_1\dots+l_a}$. Write the
union of sequences $\{x_i\}_1^m$ and $\{y_i\}_1^p$ as a
sequence $\{u_i\}_1^{m+p}$ with $x_a=u_{l_a}, a=1,\dots, m$
and $u_{l_a+k}=y_{l_1\dots+l_a+k}, k=1,\dots, l_{a+1}-1$.
In other words:
\begin{equation}
  \mathbf{u} = \set{y_1,\dots,x_1,\dots,x_2,\dots,x_m,\dots,y_p } = \set{u_1 \dots, u_{l_1}, \dots, u_{l_2}, \dots, u_{l_m}, \dots, u_{p+m} }.
\end{equation}
It is easy to check that
\begin{equation}
  \sum_{i=1}^p l( y_i | \mathbf{x}) = - \tbinom{m+1}{2} + \sum_{a=1}^m l_a \, .
\end{equation}
In terms of $\{u\}$ the formula in Eq.~\eqref{eq:prefact} is
\begin{multline}
  \label{eq:expq}
  \frac{q^{\binom{p}{2} }}{[ p ]_{q^2}!} K^{-\nicefrac{p}{2}} F^p \Omega(x_1,\dots,x_m) = q^{2 \binom{m+1}{2} + p m + \frac{p^2}{2} }\times\\
  \times { \sum_{\substack{\frac{N}{2}-\frac{1}{2} > u_1 > \dots > u_{p+m} > - \frac{N}{2}+\frac{1}{2} \\ u_{l_a} = x_a}}} q^{\sum_{i=1}^{m+p} u_i - \sum_{a=1}^m \left( u_{l_a} + 2 l_a \right)} \Omega(u_1, \dots, u_{p+m}) \, .
\end{multline}

Now we want to pass to the limit $N\to \infty, p=\frac{N}{2}-m$ with finite $m$.
In terms of variables
\begin{equation}
  \delta_i = u_i - \tfrac{1}{2} + i  \, ,
\end{equation}
the formula for $B_p^{(N)}(\mathbf{z})$ reads
\begin{multline}
  B^{(N)}_p(\mathbf{z}) =   \frac{q^{\binom{p}{2}}}{[ p ]_{q^2}!}
  K^{-\nicefrac{p}{2}} F^{p}\sum_{- \frac{N}{2}+\frac{1}{2} \le x_1 < \dots < x_m \le \frac{N}{2}-\frac{1}{2}} \lambda^{(N)}(\mathbf{x}|\mathbf{z})\Omega(\mathbf{x})\\
  = \sum_{ \frac{N}{2} \ge \delta_1 \ge \dots \ge \delta_{p+m} \ge 0} \left[
    q^{\binom{m+1}{2}+\sum_{i=1}^{p+m}\delta_i} \sum_{1\le l_1 <
      \dots < l_m \le p+m}
    \frac{\lambda^{(N)}(\delta_{l_1} - l_1 + \frac{1}{2},\dots,\delta_{l_m} - l_m + \frac{1}{2} |\mathbf{z})}
    {q^{\sum_{a=1}^m(\delta_{l_a}+ l_a)}} \right] \times \\
  \times \Omega(\delta_{1}-\tfrac{1}{2},\dots,\delta_{p+m} - p - m+
  \tfrac{1}{2}) \, .
\end{multline}
Now, set $p=\frac{N}{2}-m$. In this case, the summation in
the formula above is taken over integer partitions
$\delta \in \mathscr{S}_N$  corresponding to Young diagrams which can be inscribed into a  $N/2\times N/2$ square. We can rewrite this formula as
\begin{equation}
  \label{eq:Finite-Chain-Half-Full}
  B^{(N)}_{N/2-m}(\mathbf{z}) = \sum_{\delta \in \mathscr{S}_N} \left[
    q^{\binom{m+1}{2}+\abs{\delta}} \sum_{1\le l_1 <
      \dots < l_m \le \frac{N}{2}}
    \frac{\lambda^{(N)}(\delta_{l_1} - l_1 + \frac{1}{2},\dots,\delta_{l_m} - l_m + \frac{1}{2} |\mathbf{z})}
    {q^{\sum_{a=1}^m(\delta_{l_a}+ l_a)}} \right] \Omega_{DW}(\delta) \, ,
\end{equation}
where $\abs{\delta} = \delta_1 + \dots + \delta_r$ and $\Omega_{DW}(\delta) =
\Omega(\delta_{1}-\tfrac{1}{2},\dots,\delta_{N/2} - \frac{N}{2} +
\tfrac{1}{2})$.
Taking the limit $N\to \infty$ in this expression is premature. In
order to get a convergent series we have to rearrange this sum.

First consider
the one--particle Bethe state:
\begin{equation}
  B^{(N)}_{N/2-1} (z) = \sum_{\delta \in \mathscr{S}_N} q^{1+ \abs{\delta}} \sum_{l=1}^{N/2} \frac{\lambda^{(N)}(\delta_l -l+\frac{1}{2}|z)}{q^{\delta_l + l}} \Omega_{DW} (\delta) \, ,
\end{equation}
This expression can be written as:
\begin{equation}
  B^{(N)}_{N/2-1} (z) = \sum_{\delta \in \mathscr{S}_N} q^{1+ \abs{\delta}} \left[ \sum_{l=1}^{r} \frac{\lambda^{(N)}(\delta_l -l+\frac{1}{2}|z)}{q^{\delta_l + l}} +  \sum_{l=r+1}^{N/2} \frac{\lambda^{(N)}( -l+\frac{1}{2}|z)}{q^{ l}} \right] \Omega_{DW} (\delta) \, ,
\end{equation}
where we can use the explicit form of $\lambda^{(N)}(x|z)$ (see
App.~\ref{sec:bethe_fin}) to evaluate the second sum:
\begin{equation}
  \sum_{l=r+1}^{N/2} \frac{\lambda^{(N)}( -l+\frac{1}{2}|z)}{q^{ l}} = \frac{\frac{1}{q} \lambda^{(N)}(-r + \frac{1}{2}|z) - \lambda^{(N)}(-r - \frac{1}{2}| z)}{q^r \left( z+\frac{1}{z}-q - \frac{1}{q}\right)} \, .
\end{equation}
Now we can then take the limit $N \to \infty$  and and we obtain
the formula (\ref{eq:B-result}) for $m=1$:
\begin{equation}
  B_\infty (z) = \sum_{\delta \in \mathscr{S}} q^{1+ \abs{\delta}} \left[ \sum_{l=1}^{r} \frac{\lambda(\delta_l -l|z)}{q^{\delta_l + l}} + \frac{\frac{1}{q} \lambda(-r|z) - \lambda(-r-1 | z)}{q^r \left( z+\frac{1}{z}-q - \frac{1}{q} \right)} \right] \Omega_{DW} (\delta) \, ,
\end{equation}
Here $\mathscr{S}$ is the set of integer partitions. Using the
explicit expression for $\lambda(x|z) = z^x$ we can further rewrite
$B_\infty$ as
\begin{equation}\label{m=1}
  B_\infty (z) = \sum_{\delta \in \mathscr{S}} q^{1+ \abs{\delta}} \left[ \sum_{l=1}^{r} \frac{\lambda(\delta_l -l|z)}{q^{\delta_l + l}} - \frac{\left( z q \right)^{-r}}{1- q z}\right] \Omega_{DW} (\delta) \, .
\end{equation}
The identity
\begin{equation}
\label{eq:Divergent-Sum}
  \sum_{l=1}^{\infty} \frac{\lambda(\delta_l -l|z)}{q^{\delta_l + l}}  = \sum_{l=1}^{r} \frac{\lambda(\delta_l -l|z)}{q^{\delta_l + l}} + \sum_{l=r+1}^{\infty} \frac{1}{\left(z q\right)^{l}} = \sum_{l=1}^{r} \frac{\lambda(\delta_l -l|z)}{q^{\delta_l + l}} - \frac{\left( z q \right)^{-r}}{1- q z} \, ,
\end{equation}
which holds for $q>1$ demonstrates that the formula (\ref{m=1})
can be obtained from (\ref{eq:Finite-Chain-Half-Full}) by taking the limit in the coefficient functions for $q>1$ and then by
continuing to $q<1$ after rearranging the sum.

The same reasoning can be repeated in the general case $m \ge 1$ and
leads to the following expression for $B_{\infty}(\mathbf{z})$:
\begin{equation}
  B_{\infty}(\mathbf{z}) = q^{\binom{m+1}{2}} \sum_{\delta \in \mathscr{S}} q^{\abs{\delta}} \sum_{1 \leq l_1 < \dots l_m < \infty} \frac{\lambda ( \delta_{l_1} - l_1, \dots, \delta_{l_m} - l_m |\mathbf{z}) }{q^{\sum_{a=1}^m \delta_{l_a}+l_a}} \Omega_{DW}(\delta) \, ,
\end{equation}
where for $q<1$ each divergent geometric series is to be understood in the sense
of Eq.~(\ref{eq:Divergent-Sum}).

\subsection{Scalar products}
\label{sec:corr}

The scalar products between two eigenstates can be evaluated
by considering the commutation relations of the $U_q(sl_2)$
algebra (see Appendix~\ref{sec:UqSU2}). The scalar product is given by
\begin{equation}
  ( B_{\infty} ( \mathbf{z}^\prime) , B_\infty (\mathbf{z})) =  (B(\mathbf{z}^\prime), B(\mathbf{z})) \prod_{k=1}^\infty \frac{1}{1-q^{2k}} \, ,
\end{equation}
\emph{i.e.} a universal factor times the scalar product of the two primary Bethe states.

\bigskip

To prove this identity consider a chain of length $N$. Let us evaluate the scalar product
\begin{equation}
  ( B^{(N)}_{p^\prime} ( \mathbf{z}^\prime) , B^{(N)}_p (\mathbf{z})) =
  \frac{q^{\binom{p}{2}+\binom{p^\prime}{2}}}{[p]_{q^2}!
    [p^\prime]_{q^2}!}( K^{-\nicefrac{p^\prime}{2}}F^{p^\prime}B^{(N)}
  ( \mathbf{z}^\prime) , K^{-\nicefrac{p}{2}}F^pB^{(N)}(\mathbf{z}))
  \, .
\end{equation}
First notice that
\begin{align}
  K^\dagger = K \, , && \left( F^\dagger \right)^{p^\prime} = q^{-{p^\prime}^2}K^{p^\prime}
  E ^{p^\prime} \, ,
\end{align}
and
\begin{equation}
  F^m K^n = q^{2 m n} K^n F^m \, .
\end{equation}
The elements $E$ and $F$ satisfy a well--known identity first
derived by V. Kac:
\begin{equation}
E^{p^\prime}F^p=F^{p-p^\prime}
\frac{[p]_{q^2}!}{[p-p^\prime]_{q^2}!}K^{-p^\prime}q^{ p^\prime}
\prod_{k=1}^{p^\prime}\frac{1- K^2q^{ - 2 p + 2 k}}{1-q^{2}} + \text{ terms proportional to powers of $E$} \, .
\end{equation}
Here $p\geq p^\prime$. If $p<p^\prime$, every term is proportional to some power of $E$ (from the left).

By construction, the Bethe states are highest weight states and $E B(\mathbf{z})= 0$. Therefore all terms proportional to $E$ with vanish in the scalar product and we obtain
\begin{equation}
( B^{(N)}_{p^\prime} ( \mathbf{z}^\prime) , B^{(N)}_p (\mathbf{z}))=
\delta_{p,p^\prime} \frac{q^{2 p^2}}{[p]_{q^2}!} (B^{(N)} ( \mathbf{z}^\prime) , K^{-p} \prod_{k=1}^{p}\frac{1- K^2q^{ - 2 p + 2 k}}{1-q^{2}}B^{(N)}(\mathbf{z}))\,.
\end{equation}
In the special case $N = 2 \left( p + m \right)$ the expression
simplifies and reads:
\begin{equation}
  (B_p^{2 \left( p+m \right)}(\mathbf{z^\prime}), B^{2 \left( p + m \right)}_{p^\prime} (\mathbf{z}) ) = \delta_{p, p^\prime} (B^{2 \left( p+m \right)}(\mathbf{z}^\prime), B^{2 \left( p+m \right)}(\mathbf{z})) \qbinom{2p}{p}_{q^2} \, .
\end{equation}
where $\qbinom{2p}{p}_{q^2} $ is the $q$--binomial coefficient.
We are now in the position of taking the $p \to \infty $ limit and find:
\begin{equation}
  (B_\infty(\mathbf{z^\prime}), B_\infty (\mathbf{z}) )= \lim_{p \to \infty} (B(\mathbf{z}^\prime), B(\mathbf{z})) \qbinom{2p}{p}_{q^2} = (B(\mathbf{z}^\prime), B(\mathbf{z})) \prod_{k=1}^\infty \frac{1}{1-q^{2k}} \, .
\end{equation}

\section{Time evolution for the \textsc{asep}}\label{sec:asep}

Having obtained the complete set of eigenstates for the \textsc{xxz} chain
with domain wall boundary conditions puts us in the position to clearly
state (if not solve explicitly) the problem of the time
evolution of the asymmetric exclusion process. Consider a system that
at initial time is in the configuration $\Omega^{(0)} = \Omega_{DW}$,
where all the spins at negative positions are up and the others are
down (empty partition state), which evolves with transition rate
$W$. After a time $t$, the system will be in the state
$\Omega^{t}_{DW}$, given by
\begin{equation}
  \Omega^{t}_{DW} = e^{t W} \Omega_{DW} \, .
\end{equation}
In Sec.~\ref{sec:asep} we found that $W$ is related to the
\textsc{xxz} Hamiltonian by the similarity transformation $W = -A U
\HH U^{-1}$. This means that the time evolution operator can be written
as
\begin{equation}
  e^{t W} = U \ e^{-t A \HH} U^{-1} \, .
\end{equation}
The exponential of $\HH$ can be decomposed in terms of the eigenvectors given in Sec.~\ref{sec:spectrum} using the spectral decomposition theorem:
\begin{multline}
  \Omega_{DW}^t = \prod_{k=1}^\infty \left( 1 - q^{2k} \right) U \left[ (B_\infty, \Omega_{DW}) B_\infty + \sum_{m=1}^\infty \int \di \mathbf{z} \, \frac{(B_\infty(\mathbf{z}), \Omega_{DW})}{\norm{B(\mathbf{z})}^2} e^{-t A \mathcal{E} (\mathbf{z})} B_{\infty} (\mathbf{z}) \right]  \\
  = \prod_{k=1}^\infty \left( 1 - q^{2k} \right) \left[ \Omega_\infty + \sum_{m=1}^\infty \int \di \mathbf{z} \, \frac{(B_\infty(\mathbf{z}), \Omega_{DW})}{\norm{B(\mathbf{z})}^2} e^{-t A \mathcal{E} (\mathbf{z})} \Omega_{\infty} (\mathbf{z}) \right] \, ,
\end{multline}
where $\mathcal{E} (\mathbf{z})$ is the eigenvalue~(\ref{eq:ev}) corresponding to
$B_\infty(\mathbf{z})$, and $\Omega_\infty = B_\infty $ is the
\textsc{asep} ground state.


\section{The limit shape and matrix elements for the ground state}
\label{sec:limit-shapes}

The problem of calculating the magnetization profile of the ground
state $B_\infty(0)$ is closely related to the classical
problem of finding the limit shape of the ensemble of random
partitions, where the limit of vanishing lattice spacing and $q\to
1^-$ is taken. The limit shape is simply the integral of the kink
shaped magnetization profile of the ground state.
Here, we recall results of~\cite{Dijkgraaf:2008}, making use of the techniques developed in the present article.

The magnetization profile of a state of the spin chain labeled by a vector
$v$ is given by
\begin{equation}
  m_v (x) = \frac{(v, \sigma_x^3 v)}{(v,v)} \, .
\end{equation}
Consider a chain of length $N$. The ground state fulfills $ B^{N}_p(0) \propto \left( S^-
\right)^{p} \Omega^{(N)}$, thus the corresponding magnetization profile
reads
\begin{equation}
  m_{B_p^{(N)}(0)} (x) = \frac{ ( K^{-\nicefrac{p}{2}}F^{p}B^{(N)}
  ( \mathbf{z}^\prime) , \sigma_x^3 K^{-\nicefrac{p}{2}}F^pB^{(N)}(\mathbf{z})) }{( K^{-\nicefrac{p}{2}}F^{p}B^{(N)}
  ( \mathbf{z}^\prime) , K^{-\nicefrac{p}{2}}F^pB^{(N)}(\mathbf{z})) }\, .
\end{equation}
This can be evaluated explicitly using the commutation relation
\begin{equation}
  \comm{ \sigma^3_x , F^p } =  2 [p]_{q^2}! L_x \sum_{ y_1 < \dots < y_{p-1} } L_{y_1} \dots L_{y_{p-1}} \sigma^-_{y_1} \dots \sigma^-_{y_{p-1}} \sigma^-_x \ \sigma^3_x  \, .
\end{equation}
%
and using the fact that $\sigma_x^3 \Omega^{(N)} = \Omega^{(N)}$.

The final expression for $N=2p$ is given by
\begin{equation}
  \label{eq:magnetization-profile}
  m_{B^{(2p)}_p} (x)   = 1 - 2 \sum_{k=0}^p q^{2 k \left( x + p \right)} \frac{( q^{-2p} ; q^2)_k }{( q^{2+2p} ; q^2)_k}\, .
\end{equation}
It is easy to obtain the limit shape, when we consider first the limit $p \to \infty$ for fixed $q$, and then $q\to 1^{-}$.
Sending $p\to \infty$ we obtain
\begin{equation}
  m_{B_\infty} (x)   = 1 - 2 \sum_{k=0}^\infty \left( -1 \right)^k
  q^{2 \binom{k}{2}} q^{2 k  x } \, .
\end{equation}
Assume that the  variable $u = - 2 x \log q$ remains finite when $ q \to 1^- $. Then the magnetization profile converges to
\begin{equation}
  m_{B_\infty}(u)  = 1 - 2 \sum_{k=0}^\infty \left( - 1 \right)^k e^{-
    k u} = \frac{1 - e^u }{ 1 + e^u} \, ,
\end{equation}
and the corresponding limit shape is the antiderivative of $-
m_{B_{\infty}}(u)$ :
\begin{equation}
  \mu(u) = 2 \log \left(2 \cosh \left(\frac{u}{2} \right) \right)\, .
\end{equation}
Indeed, $\mu(u)'=-m_{B_{\infty}}(u)$, and $\mu(u)\to |u|$ as $|u|\to \infty$.

The detailed probabilistic analysis of the limit $p\to \infty$ with
$q^p$ being fixed is done in~\cite{BBE}.


\subsubsection*{Matrix elements for the (projector on) the ground state}
\label{sec:mat-el}

 Here we calculate matrix elements for the projector on
the ground state, \emph{i.e.} the probability on the ground
  state to have the $l_1$-th particle move by $\delta_1$, the $l_2$-th particle move by $\delta_2$, etc. starting from $\Omega_{DW}$.

Consider the half full sector for a chain of length $2p$. The ground
state can be written as
\begin{equation}
  B_p^{(2p)} = \sum_{-p + \frac{1}{2} \le y_p < \dots < y_1 \le p - \frac{1}{2}} q^{\sum_{i=1}^p y_i} \Omega^{(2p)}(y_1, \dots, y_p) \, .
\end{equation}
The probability for the $l$--th particle from the right to move by $\delta$ is given by
\begin{equation}
  P(l,\delta) = \frac{( B_p^{(2p)}, \delta(y_l - x) B_p^{(2p)})}{(B_p^{(2p)}, B_p^{(2p)} )} \, ,
\end{equation}
where $x = \delta - l + \frac{1}{2}$ is the position in the chain of
the $l$--th particle.
Explicitly,
\begin{multline}
  P(l, \delta) = \frac{\displaystyle{\sum_{-p + \frac{1}{2} \le y_p < \dots < y_1 \le p - \frac{1}{2}} \delta(y_l - x )q^{2\sum_{i=1}^p y_i}}}{\displaystyle{\sum_{-p + \frac{1}{2} \le y_p < \dots < y_1 \le p - \frac{1}{2}} q^{2\sum_{i=1}^p y_i}}} =\\
  = \frac{\displaystyle{\sum_{-p + \frac{1}{2} \le y_p < \dots <
        y_{l+1} \le x - 1} \left[ q^{2\sum_{i=l+1}^{p} y_i} \right]
      q^{2x} \sum_{x + 1 \le y_{l-1} < \dots < y_1 \le p -
        \frac{1}{2}} \left[ q^{2\sum_{i=1}^{l-1} y_i}
      \right]}}{\displaystyle{\sum_{-p + \frac{1}{2} \le y_p < \dots <
        y_1 \le p - \frac{1}{2}} q^{2\sum_{i=1}^p y_i}}} \, .
\end{multline}
Using the identity
\begin{equation}
  \sum_{a+1\le y_m < \dots < y_1 \le b} q^{2 \sum_{i=1}^m y_i} = q^{2 \binom{m+1}{2} + 2 a m} \qbinom{b-a}{m}_{q^2} \, ,
\end{equation}
the probability is
\begin{equation}
  P(l, \delta) = q^{2 l \delta} \frac{\qbinom{p + \delta - l}{\delta}_{q^2}\qbinom{p-\delta+l-1}{l-1}_{q^2}}{\qbinom{2p}{p}_{q^2}} \, ,
\end{equation}
and in the infinite case we find
\begin{equation}
  P( l, \delta ) =  q^{2 l \delta} \frac{(q^2;q^2)_\infty }{(q^2;q^2)_{l-1} (q^2;q^2)_{\delta}} \, .
\end{equation}
Note that the magnetization profile in Eq.~(\ref{eq:magnetization-profile}) can be obtained starting with these probabilities and summing for $l \in \setN$, at fixed position $x$.

\bigskip

In a similar fashion one can calculate the joint probability of having
the $l_1$-th particle move by $\delta_1$, the $l_2$-th particle move
by $\delta_2$ and so on. It is given by
\begin{equation}
  P ( \set{l_i}_{i=1}^m, \set{\delta_i}_{i=1}^m ) = \frac{\displaystyle{\sum_{-p + \frac{1}{2} \le y_p < \dots < y_1 \le p - \frac{1}{2}} \left[ \prod_{i=1}^m \delta(y_i - x_i ) \right]q^{2\sum_{i=1}^p y_i}}}{\displaystyle{\sum_{-p + \frac{1}{2} \le y_p < \dots < y_1 \le p - \frac{1}{2}} q^{2\sum_{i=1}^p y_i}}} \, .
\end{equation}
It can be expressed as
\begin{multline}
  P(\set{l_i}_{i=1}^n,\set{\delta_i}_{i=1}^n) = q^{2\delta_1 l_1 + 2 \sum_{i=1}^{n-1} \delta_{i+1} \left( l_{i+1} - l_{i} \right)} \frac{\qbinom{p + \delta_1 - l}{\delta_1}_{q^2}\qbinom{p-\delta_n+l_n-1}{l_n-1}_{q^2}}{\qbinom{2p}{p}_{q^2}}\times\\
  \times \prod_{i=1}^{n-1} \qbinom{\delta_{i} - \delta_{i+1} + l_{i+1} - l_{i} - 1}{\delta_{i}-\delta_{i+1}}_{q^2} \, ,
\end{multline}
and in the $p \to \infty $ limit,
\begin{multline}
  P(\set{l_i}_{i=1}^n,\set{\delta_i}_{i=1}^n) = q^{2\delta_1 l_1 + 2 \sum_{i=1}^{n-1} \delta_{i+1} \left( l_{i+1} - l_{i} \right)} \frac{(q^2;q^2)_\infty }{(q^2;q^2)_{l_n-1} (q^2;q^2)_{\delta_1}} \times\\
  \times\prod_{i=1}^{n-1} \qbinom{\delta_{i} - \delta_{i+1} + l_{i+1} - l_{i} - 1}{\delta_{i}-\delta_{i+1}}_{q^2} \, .
\end{multline}
Note that this is the same as the probability of
having an integer partition with the number $\delta_1$ at
position $l_1$, the number $\delta_2$ at position $l_2$, and so on
(here it is obvious that $l_i < l_{i+1} $ and $\delta_i >
\delta_{i+1}$),
\begin{equation}
  P( [ \dots \stackrel{l_1}{\delta_1} \dots \stackrel{l_2}{\delta_2} \dots \stackrel{l_n}{\delta_n}\dots ] ) \, .
\end{equation}

\subsection*{Acknowledgements}

It is our pleasure to thank  Robbert Dijkgraaf, Kyoji Saito, and the members of the IPMU string theory group meeting for discussions.
Furthermore, D.O and S.R. would like to thank the VI. Simons Workshop
in Stony Brook for hospitality.  The research of D.O. and S.R. was
supported by the World Premier International Research Center
Initiative (WPI Initiative), MEXT, Japan.
The work of N.R. is supported in part by the Danish National
Research Foundation through the
Niels Bohr initiative, and by the NSF grant DMS-0601912.

\bibliography{DimerReferences}

\newpage

\appendix

\section{The algebra $U_q(sl_2)$}\label{sec:UqSU2}

The algebra $U_q(sl_2)$ is generated by $E,\ F,\ K,$ and $K^{-1}$ under the relations
\begin{align}
K\,K^{-1}=K^{-1}\,K=\mathbbm{1} \, , && [E,F]=\frac{K-K^{-1}}{q-q^{-1}} \, ,\\
K\,E\,K^{-1}=q^{2}E \, , && K\,F\,K^{-1}=q^{-2}F \, .
\end{align}
It is a Hopf algebra deformation of the universal enveloping algebra of $sl_2$ with the comultiplication
\begin{equation}
\Delta K=K\otimes K \, , \ \ \Delta E=E\otimes \mathbbm{1}+ K^{-1}\otimes E \, ,
\ \ \Delta F=F \otimes K +\mathbbm{1} \otimes F \, .
\end{equation}

Iterating this comultiplication, and evaluating the algebra
in its two dimensional representation gives the action of $U_q(sl_2)$
on the space of states of a spin chain of length $N$ with spin $s=\tfrac{1}{2}$:
\begin{align}\label{u-action}
  F&=\sum_y\underbrace{q^{\sigma^3} \otimes \dots \otimes q^{\sigma^3}}_{y-1}\otimes\, \sigma^-\otimes\mathbbm{1}\otimes\dots\otimes\mathbbm{1} = \sum_yL_y\,\sigma_y^- \, ,\\
  E&=\sum_y\mathbbm{1}\otimes\dots\otimes\mathbbm{1}\otimes\sigma^+\otimes \underbrace{q^{-\sigma^3}\otimes\dots \otimes q^{-\sigma^3}}_{N-y} \, ,\\
  K&=\underbrace{q^{\sigma_3 } \otimes \dots \otimes q^{\sigma_3 }}_N\,,
\end{align}
where $\sigma^\pm,\ \sigma^3$ are the Pauli matrices
\begin{align}
\sigma^+=\begin{pmatrix}
    0 & 1 \\ 0 & 0
  \end{pmatrix} \, , && \sigma^-=\begin{pmatrix}
    0 & 0 \\ 1 & 0
  \end{pmatrix} \, , && \sigma^3 = \begin{pmatrix}
     1 & 0 \\ 0 & -1 
     \end{pmatrix} \, ,
\end{align}
and
\begin{equation}
  q^{\sigma_3}=\begin{pmatrix}
    q & 0 \\ 0 & q^{-1}
  \end{pmatrix}.
\end{equation}

In our case, $q$ is not a root of unity, and therefore $({\mathbb C}^{2})^{\otimes N}$  splits into the
direct sum of irreducible highest weight representations of $U_q(sl_2)$.

\section{The Bethe ansatz for the finite \textsc{xxz} chain}
\label{sec:bethe_fin}

In this appendix we recall the expressions for the Bethe vectors for
the finite chain of length $N$ with $U_q(sl_2)$-invariant boundary
conditions.

Special boundary conditions for the \textsc{xxz} Hamiltonian when it
commutes with the action (\ref{u-action}) of $U_q(sl_2)$ on the space
of states were found in~\cite{Pasquier:1990}. Even more remarkable is
that this Hamiltonian describes an integrable system: the Bethe alsatz
for it was constructed in~\cite{Pasquier:1990}, the algebraic version
and the complete systems of commuting integrals were found
in~\cite{Sklyanin}.

The Hamiltonian found in~\cite{Pasquier:1990} has the following form:
\begin{equation}
  \label{eq:fin_xxz}
  \HH = - \frac{1}{2} \sum_{k=-N/2+1/2}^{N/2-1/2} (\sigma^1_k \sigma_{k+1}^1 +  \sigma^2_k \sigma_{k+1}^2 + \Delta \ \sigma^3_k \sigma_{k+1}^3 + \tfrac{1}{2}(q-\tfrac{1}{q}) (\sigma^3_k-\sigma^3_{k+1}) - \Delta )\, .
\end{equation}

Eigenvectors for this Hamiltonian in the subspace where
$m$ spins are down and the rest up are given by the following formulae:
\begin{equation}
  \label{eq:finite_ev}
  B^{(N)} (\mathbf{z}) = \sum_{ - \frac{N}{2} + \frac{1}{2}< x_1 < \dots < x_m < \frac{N}{2}-\frac{1}{2}}  \lambda^{(N)} ( \mathbf{x} | \mathbf{z} ) \Omega^{(N)}( \mathbf{x}) \,
\end{equation}
where $\Omega^{(N)}(\mathbf{x})$ is a state with $m$ spins down in the
positions with coordinates $\mathbf{x}$, and
\begin{equation}
  \lambda^{(N)}( \mathbf{x} | \mathbf{z}) = \sum_{w\in W} \epsilon_w A^{(N)}(z_{w_1}, \dots, z_{w_m}) z_{w_1}^{x_1}\dots z_{w_m}^{x_m} \, ,
\end{equation}
Here the sum runs over the permutations and reflections ($z_i\mapsto
z_i^{-1}$) of $\set{z_1, \dots, z_m}$, and $\epsilon_W$ is the
signature of $w$. The coefficients $A^{(N)}(z)$ are:
\begin{equation}
  A^{(N)} (z ) = \prod_{j=1}^m \beta( z_j^{-1} ) \prod_{1  \le j
    < l \le m} B(z_j^{-1}, z_l) z_l^{-1} \,
\end{equation}
where
\begin{align}
   B(z_1, z_2 ) &= \left( 1 -  2 \Delta z_2 +z_1z_2 \right) \left( 1 - 2 \Delta z_1^{-1} + z_2z_1^{-1}\right) \, ,
\end{align}
and
\begin{equation}
\beta(z) =
  \left( 1 - q \ z^{-1} \right) z^{(N+1)/2} \, .
\end{equation}
The vector (\ref{eq:finite_ev}) is an eigenvector if the numbers $z_j$
have modulus $\abs{z_j} = 1 $ and satisfy the Bethe equations:
\begin{equation}
   z_j^{2N} = \prod_{\substack{l=1\\ l \neq j}}^m\frac{B(z_j^{-1},z_l)}{B(z_j,z_l)},\quad j=1,\dots,m \, .
\end{equation}
The corresponding eigenvalue is given by
\begin{equation}
 \mathcal{E}^{(N)}(\mathbf{z})  =  \sum_{j=1}^m \left( 2 \Delta - z_j - z_j^{-1} \right) \, .
\end{equation}

It is quite remarkable that in addition to being eigenvectors
these vectors are also $U_q(sl_2)$ highest weight vectors with the weight $N-2m$ ~\cite{Pasquier:1990}:
\begin{align}
  E B^{(N)} (\mathbf{z}) = 0 \ , && K B^{(N)} (\mathbf{z})=q^{N-2m}
  B^{(N)} (\mathbf{z})\, .
\end{align}
Since the Hamiltonian (\ref{eq:fin_xxz}) commutes with the action of
$U_q(sl_2)$, the vectors $F^pB^{(N)} (\mathbf{z})$ are
also eigenvectors.

\end{document}